\documentclass[pra,aps,showpacs,twocolumn,superscriptaddress]{revtex4}

\usepackage{amsmath,amsfonts,amssymb,caption,hyperref,color,epsfig,graphics,graphicx,latexsym,mathrsfs,revsymb,theorem,url,verbatim,epstopdf}

\hypersetup{colorlinks,linkcolor={blue},citecolor={blue},urlcolor={red}}

\usepackage{hyperref}

\makeatletter

\newcommand{\Rmnum}[1]{\expandafter\@slowromancap\romannumeral #1@}
\makeatother
\newtheorem{definition}{Definition}
\newtheorem{proposition}[definition]{Proposition}

\newtheorem{Theorem}[definition]{Theorem}

\newtheorem{conjecture}[definition]{Conjecture}

\newtheorem{remark}[definition]{Remark}
\newtheorem{example}[definition]{Example}
\newtheorem{question}[definition]{Question}

\def\squareforqed{\hbox{\rlap{$\sqcap$}$\sqcup$}}
\def\qed{\ifmmode\squareforqed\else{\unskip\nobreak\hfil
		\penalty50\hskip1em\null\nobreak\hfil\squareforqed
		\parfillskip=0pt\finalhyphendemerits=0\endgraf}\fi}
\def\endenv{\ifmmode\;\else{\unskip\nobreak\hfil
		\penalty50\hskip1em\null\nobreak\hfil\;
		\parfillskip=0pt\finalhyphendemerits=0\endgraf}\fi}
\newenvironment{proof}{\noindent \textbf{{Proof.~} }}{\qed}
\def\Dbar{\leavevmode\lower.6ex\hbox to 0pt
	{\hskip-.23ex\accent"16\hss}D}
\makeatletter
\def\url@leostyle{%
	\@ifundefined{selectfont}{\def\UrlFont{\sf}}{\def\UrlFont{\small\ttfamily}}}
\makeatother
\def\bcj{\begin{conjecture}}
	\def\ecj{\end{conjecture}}
\def\bcr{\begin{corollary}}
	\def\ecr{\end{corollary}}
\def\bd{\begin{definition}}
	\def\ed{\end{definition}}
\def\bea{\begin{eqnarray}}
\def\eea{\end{eqnarray}}
\def\bem{\begin{enumerate}}
	\def\eem{\end{enumerate}}
\def\bex{\begin{example}}
	\def\eex{\end{example}}
\def\bim{\begin{itemize}}
	\def\eim{\end{itemize}}
\def\bl{\begin{lemma}}
	\def\el{\end{lemma}}
\def\bma{\begin{bmatrix}}
	\def\ema{\end{bmatrix}}
\def\bpf{\begin{proof}}
	\def\epf{\end{proof}}
\def\bpp{\begin{proposition}}
	\def\epp{\end{proposition}}
\def\bqu{\begin{question}}
	\def\equ{\end{question}}
\def\br{\begin{remark}}
	\def\er{\end{remark}}
\def\bt{\begin{theorem}}
	\def\et{\end{theorem}}

\def\btb{\begin{tabular}}
	\def\etb{\end{tabular}}

\newcommand{\nc}{\newcommand}

\def\a{\alpha}
\def\b{\beta}

\def\r{\rho}
\def\s{\sigma}

\def\o{\omega}

\nc{\bbA}{\mathbb{A}} \nc{\bbB}{\mathbb{B}} \nc{\bbC}{\mathbb{C}}
\nc{\bbD}{\mathbb{D}} \nc{\bbE}{\mathbb{E}} \nc{\bbF}{\mathbb{F}}
\nc{\bbG}{\mathbb{G}} \nc{\bbH}{\mathbb{H}} \nc{\bbI}{\mathbb{I}}
\nc{\bbJ}{\mathbb{J}} \nc{\bbK}{\mathbb{K}} \nc{\bbL}{\mathbb{L}}
\nc{\bbM}{\mathbb{M}} \nc{\bbN}{\mathbb{N}} \nc{\bbO}{\mathbb{O}}
\nc{\bbP}{\mathbb{P}} \nc{\bbQ}{\mathbb{Q}} \nc{\bbR}{\mathbb{R}}
\nc{\bbS}{\mathbb{S}} \nc{\bbT}{\mathbb{T}} \nc{\bbU}{\mathbb{U}}
\nc{\bbV}{\mathbb{V}} \nc{\bbW}{\mathbb{W}} \nc{\bbX}{\mathbb{X}}
\nc{\bbZ}{\mathbb{Z}}

\nc{\bA}{{\bf A}} \nc{\bB}{{\bf B}} \nc{\bC}{{\bf C}}
\nc{\bD}{{\bf D}} \nc{\bE}{{\bf E}} \nc{\bF}{{\bf F}}
\nc{\bG}{{\bf G}} \nc{\bH}{{\bf H}} \nc{\bI}{{\bf I}}
\nc{\bJ}{{\bf J}} \nc{\bK}{{\bf K}} \nc{\bL}{{\bf L}}
\nc{\bM}{{\bf M}} \nc{\bN}{{\bf N}} \nc{\bO}{{\bf O}}
\nc{\bP}{{\bf P}} \nc{\bQ}{{\bf Q}} \nc{\bR}{{\bf R}}
\nc{\bS}{{\bf S}} \nc{\bT}{{\bf T}} \nc{\bU}{{\bf U}}
\nc{\bV}{{\bf V}} \nc{\bW}{{\bf W}} \nc{\bX}{{\bf X}}
\nc{\bZ}{{\bf Z}}

\nc{\cA}{{\cal A}} \nc{\cB}{{\cal B}} \nc{\cC}{{\cal C}}
\nc{\cD}{{\cal D}} \nc{\cE}{{\cal E}} \nc{\cF}{{\cal F}}
\nc{\cG}{{\cal G}} \nc{\cH}{{\cal H}} \nc{\cI}{{\cal I}}
\nc{\cJ}{{\cal J}} \nc{\cK}{{\cal K}} \nc{\cL}{{\cal L}}
\nc{\cM}{{\cal M}} \nc{\cN}{{\cal N}} \nc{\cO}{{\cal O}}
\nc{\cP}{{\cal P}} \nc{\cQ}{{\cal Q}} \nc{\cR}{{\cal R}}
\nc{\cS}{{\cal S}} \nc{\cT}{{\cal T}} \nc{\cU}{{\cal U}}
\nc{\cV}{{\cal V}} \nc{\cW}{{\cal W}} \nc{\cX}{{\cal X}}
\nc{\cZ}{{\cal Z}}

\nc{\hA}{{\hat{A}}} \nc{\hB}{{\hat{B}}} \nc{\hC}{{\hat{C}}}
\nc{\hD}{{\hat{D}}} \nc{\hE}{{\hat{E}}} \nc{\hF}{{\hat{F}}}
\nc{\hG}{{\hat{G}}} \nc{\hH}{{\hat{H}}} \nc{\hI}{{\hat{I}}}
\nc{\hJ}{{\hat{J}}} \nc{\hK}{{\hat{K}}} \nc{\hL}{{\hat{L}}}
\nc{\hM}{{\hat{M}}} \nc{\hN}{{\hat{N}}} \nc{\hO}{{\hat{O}}}
\nc{\hP}{{\hat{P}}} \nc{\hR}{{\hat{R}}} \nc{\hS}{{\hat{S}}}
\nc{\hT}{{\hat{T}}} \nc{\hU}{{\hat{U}}} \nc{\hV}{{\hat{V}}}
\nc{\hW}{{\hat{W}}} \nc{\hX}{{\hat{X}}} \nc{\hZ}{{\hat{Z}}}

\nc{\hn}{{\hat{n}}}

\def\diag{\mathop{\rm diag}}

\def\max{\mathop{\rm max}}
\def\min{\mathop{\rm min}}

\def\tr{\mathop{\rm Tr}}

\newcommand{\bra}[1]{\langle#1|}
\newcommand{\ket}[1]{|#1\rangle}

\newcommand{\norm}[1]{\lVert#1\rVert}

\def\Dbar{\leavevmode\lower.6ex\hbox to 0pt
	{\hskip-.23ex\accent"16\hss}D}

\begin{document}
	\title{A family of separability criteria and lower bounds of concurrence}
	
	\author{Xian Shi}\email[]
	{shixian01@gmail.com}
	\affiliation{College of Information Science and Technology,
		Beijing University of Chemical Technology, Beijing 100029, China}
	
		\author{Yashuai Sun}
	\affiliation{College of Information Science and Technology,
		Beijing University of Chemical Technology, Beijing 100029, China}

	%
	
	
	
	\date{\today}
	
	
	\begin{abstract}
	The problem on detecting the entanglement of a bipartite state is  significant in quantum information theory. In this article, we apply the Ky Fan norm to the revised realignment matrix of a bipartite state. Specifially, we consider a family of separable criteria for bipartite states, and present when the density matrix corresponds to a state is real, the criteria is equivalent to the enhanced realignment criterion. Moreover, we present analytical lower bounds of concurrence and the convex-roof extended negativity for arbitrary dimensional systems.
 	\end{abstract}
 \maketitle
	\section{Introduction}
\indent	Entanglement is one of the essential features in quantum mechanics when comparing with the classical physics \cite{horodecki2009quantum,plenio2014introduction}. It plays key roles in quantum information processing, such as, quantum cryptography \cite{ekert1991quantum}, teleportation \cite{bennett1993teleporting} and superdense coding \cite{bennett1992communication}. \par
 One of the most important problems in quantum information theory is how to distinguish separable and entangled states. If a quantum state $\r_{AB}$ of a bipartite system can be written as a convex combination of product states,
 \begin{align*}
 \r_{AB}=\sum_i p_i\rho_A^i\otimes\r_B^i,
 \end{align*}here $\{p_i\}$ is a probability distribution, $\rho_A^i$ and $\rho_B^i$ are states of subsystems $A$ and $B$, respectively, then it is separable, otherwise, it is entangled. The above problem is completely solved for $2\otimes2$ and $2\otimes3$ systems by the Peres-Horodecki criterion: a bipartite state $\r_{AB}$ is separable if and only if it is positive partial transpose (PPT), $i.$ $e.$, $(id\otimes T)(\rho_{AB})\ge 0$ \cite{peres1996separability}. However, the problem is NP-hard for arbitrary dimensional systems \cite{gurvits2003classical}.  In the past twenty years, there are several other prominent criterions. The computable cross norm or realignment criterion (CCNR) criterion is proposed by Rudolph \cite{rudolph2005further} and Chen and Wu \cite{chen2003matrix}. In 2006, the authors proposed the local uncertainty relations (LURs) and showed that the LURs is stronger than the CCNR criterion \cite{guhne2006entanglement}. In 2007, the author proposed a criterion which is based on Bloch representations \cite{de}. Then Zhang $et$ $al.$ presented the enhanced realignment criterion \cite{zhang2008entanglement}. In 2015, the authors proposed an impoved CCNR criterion where they showed it is stronger than the CCNR criterion \cite{shen2015separability}. In 2018, Shang $et$ $al.$ presented a sufficient condition for the separability of a bipartite state, which is called ESIC criterion \cite{shang2018enhanced}. Recently, Sarbicki $et$ $al.$ proposed a family of separability criteria which are based on the bloch representation of a state \cite{sarbicki2020family}. Subsequently, they showed that the detection power of the criteria is equivalent to the enhanced realignment criterion \cite{sarbicki2020enhanced}.\par 
 To quantify the entanglement is the other important problem in quantum entanglement theory, various entanglement measures are proposed in the past years \cite{wootters1998entanglement,vedral1998entanglement,vidal1999robustness,terhal2000schmidt,vidal2000entanglement,wei2003geometric,christandl2004squashed}. Concurrence \cite{wootters1998entanglement} and the convex-roof extended negativity (CREN) \cite{lee2003convex} are two of the most used among the measures under the convex roof extended method. However, the two measures are difficult to compute for higher dimensional systems.  Lots of results have been made on the lower bounds of concurrence and CREN \cite{chen2005concurrence,brandao2005quantifying,de2007lower,Chen2012,li2020improved}.  In \cite{de2007lower}, de Vicente presented analytical lower bounds of concurrence in terms of LUR and correlation matrix separability criteria \cite{de2007lower}. Recently, Li $et$ $al.$ improved the lower bounds of concrrence and CREN based on Bloch representations \cite{li2020improved}.\par
 In this work, we propose a class of separablity criteria based on the Ky Fan norm of the realignment matrix of a state. Moreover, we show when a bipartite state corresponds to a real matrix, the criterion is as strong as the enhanced realignment criterion. At last, we present lower bounds for both concurrence and CREN under the methods here.\par 

\section{Separability Crierion for Bipartite States}
\indent In this section, we first introduce some knowledge needed.  Assume $A=[a_{ij}]\in M_{m\times n}(\mathbb{C})$ is a matrix, $vec(A)$ is defined as
\begin{align}
vec(A)=(a_{11},a_{21},\cdots,a_{m1},a_{12},\cdots,a_{m2},\cdots,a_{mn})^{T},
\end{align}
here $T$ means transposition. Let $Z$ be an $m\times m$ block matrix with block size $n\times n$. Then the realignment operator $\mathcal{R}$ changes $Z$ into a new matrix with size $m^2\times n^2,$
$$
\mathcal{R}(Z)\equiv \begin{pmatrix}
vec(Z_{1,1})^T\\\cdots\\vec(Z_{m,1})^T\\\cdots\\vec(Z_{1,m})^T\\\cdots\\vec(Z_{m,m})^T
\end{pmatrix}.
$$The CCNR criterion \cite{chen2003matrix,rudolph2005further} shows that any separable state $\rho_{AB}$ satisfies 
\begin{align}
||\mathcal{R}(\rho)||_1\le 1. \label{ccnr}
\end{align}
Then based on the realignment of $\rho_{AB}-\rho_A\otimes\rho_B$, Zhang $et$ $al.$ showed that for any separable state $\rho_{AB}$, the following inequality is valid,
\begin{align}
\norm{\mathcal{R}(\rho_{AB}-\rho_A\otimes\rho_B)}_1\le \sqrt{1-tr\rho_A^2}\sqrt{1-\tr\rho_B^2},\label{zh}
\end{align}
it is stronger than the CCNR criterion (\ref{ccnr}) \cite{zhang2008entanglement}.\par
By using some parameters and the reduced density matrices of a bipartite state $\rho_{AB},$ Shen $et$ $al.$ \cite{shen2015separability} constructed
\begin{align}
\mathcal{N}_{\beta,l}^{G}(\rho)=\begin{pmatrix}
G&\beta \omega_l(\rho_B)^T\\
\beta\omega_l(\rho_A)&\mathcal{R}(\rho)
\end{pmatrix},
\end{align}
here $G-\b^2 E_{l\times l}$ is positive semidefinte, $\b\in \mathbb{R}$, and $\o_l(X)$ means 
\begin{align}
\o_l(X)=(\underbrace{vec(X),\cdots, vec(X)}_{\textit{$l$ columns}}).
\end{align} 
There they showed that when $G-\b^2 E_{l\times l}\ge 0$, then a separable state $\r$ satisfies 
\begin{align}
\norm{\mathcal{N}_{\beta,l}^{G}(\rho)}_1\le 1+tr(G).\label{s}
\end{align}
\indent In this manuscript, we denote $\mathcal{M}_{\alpha,\beta}(\rho_{AB})$ on a bipartite state $\r_{AB}$ as
\begin{align}
\mathcal{M}_{\alpha,\beta}(\rho_{AB})=\begin{pmatrix}
	\alpha\beta& \alpha vec(\rho_B)^T\\
	\b vec(\rho_A)& \mathcal{R}(\rho_{AB})
\end{pmatrix},
\end{align}
here $\a,\b\in \mathbb{R}$, $\rho_A$ and $\rho_B$ are redeced density matrices of the $A$ and $B$ system, respectively.  As Ky Fan norm is common used in matrix analysis \cite{bhatia2013matrix}, we use Ky Fan norm in this manuscript. The Ky Fan norm of a matrix $A_{m\times n}$ is defined as the sum of all singular values $d_i$, that is, 
 \begin{align*}
 \norm{A}_{KF}=\sum_i^{\min(n,m)}d_i= tr\sqrt{A^{\dagger} A}.
 \end{align*} \par
 \begin{Theorem}\label{t1}
 	Assume $\rho_{AB}$ is a separable state, when $\a,\b\in\mathbb{R},$ 
 	\begin{align*}
 	\norm{\mathcal{M}_{\a,\b}(\rho_{AB})}_{KF}\le \sqrt{(\a^2+1)(\b^2+1)}
 	\end{align*}
 \end{Theorem}
\begin{proof}
	As $\rho_{AB}$ is separable, it can be written as 
	\begin{align}
	\rho_{AB}=\sum_i p_i\rho_A^i\otimes\rho_B^i,\label{t}
	\end{align}
	here $p_i\in[0,1],$ $\sum_ip_i=1$, $\rho_A^i$ and $\rho_B^i$ are pure states of $A$ and $B$ systems, respectively. Next 
	\begin{align*}
	\norm{\mathcal{M}_{\a,\b}(\rho_{AB})}_{KF}=&\norm{\sum_ip_i\begin{pmatrix}
	\a\b&\a vec(\rho^i_B)^T\\
	\b vec(\rho^i_A)& vec(\rho^i_A) vec(\rho_B^i)^T
	\end{pmatrix}}_{KF}\nonumber\\
	=&\norm{\sum_ip_i\begin{pmatrix}
		\a\\
		vec(\rho^i_A)
		\end{pmatrix}\begin{pmatrix}
		\b& vec(\rho^i_B)^T
		\end{pmatrix}}_{KF}\nonumber\\
	\le&\sum_i p_i\norm{\begin{pmatrix}
		\a\\
		vec(\rho^i_A)
		\end{pmatrix}}\norm{\begin{pmatrix}
		\b\\ vec(\rho^i_B)
		\end{pmatrix}}\\\le&
	\sum_i p_i\sqrt{(\alpha^2+1)(\b^2+1)}\\=&\sqrt{(\alpha^2+1)(\b^2+1)}.
	\end{align*}
Here we use $\rho_A^i$ and $\rho_B^i$ are pure states.	Hence, we finish the proof.
\end{proof}\par
Obviously, when $\a=\b=0,$ the formula $(\ref{t})$ is the CCNR criterion. Hence the criterion is stronger than the CCNR criterion.\par
Next we present an example of a $3\otimes3$ state which cannot be detected by CCNR criterion. 
\begin{example}
\emph{	Here we consider the separability on the mixture of the chessboard state \cite{bruss2000construction} with white noise. \\
	The chessboard states is defined as follows,}
\begin{align*}
\rho=&\frac{1}{N}\sum_i\ket{V_i}\bra{V_i},\nonumber\\
\ket{V_1}=&\ket{m,0,s;0,n,0;0,0,0},\\\ket{V_2}=&\ket{0,a,0;b,0,c;0,0,0}\\
\ket{V_3}=&\ket{n^{*},0,0;0,-m^{*},0,t,0,0},\\\ket{V_4}=&\ket{0,b^{*},0;-a^{*},0,0;0,d,0}.
\end{align*}\emph{
Here $N$ is a normalized factor. To make this class of states be PPT, here we assume $m,s,n,a,b,c,t$ and $d$ are real parameters, $s=\frac{ac}{n}$, $t=\frac{ad}{m}$. The mixture of the chessboard state with the white noise is $$\r_p=p\r+(1-p)\frac{I_3\otimes I_3}{9}.$$
Next we randomly choose \begin{align*}
a=&0.33;\hspace{3mm}
b=-0.109;\hspace{3mm}
c=-0.65,\hspace{2mm}p=0.9;\\
m=&0.469;\hspace{3mm}
n=-0.3161;\hspace{3mm}
d=0.8560.
\end{align*}
Through computation, we have $\r_{p}$ is PPT, and this state cannot be detected by CCNR criterion\cite{chen2003matrix,rudolph2005further}. Next when we take $\a=250$ and $\b=240,$ $\norm{\mathcal{M}_{\a,\b}(\rho_{AB})}_{KF}- \sqrt{(\a^2+1)(\b^2+1)}=0.0027>0,$ due to the Theorem \ref{t1}, $\r_{p}$ is entangled.}
\end{example}\par
\indent Recently, the authors showed that the criterion proposed in \cite{sarbicki2020family,sarbicki2020enhanced} is as strong as the enhanced realignment criterion (\ref{zh}), here we show that the detection power of  Theorem \ref{t1} is the same as the power of the enhanced realignment criterion when the bipartite state $\rho_{AB}$ corresponds to a real matrix. 
\begin{Theorem}
	Assume $\rho_{AB}$ corresponds to a real matrix, it satisfies the enhanced realignment criterion if and only if it satisfies Theorem \ref{t1} for all $\a,\b\in \mathbb{R}^{+}.$
\end{Theorem}
\begin{proof}
	$\Longrightarrow$:\par As
	\begin{align}
	\mathcal{M}_{\a,\b}(\rho_{AB})=\mathcal{M}_{\a,\b}(\r_{A}\otimes\r_B)+\mathcal{L}_{\a,\b}(\r_{AB}),
	\end{align}
	here 
	\begin{align*}
	\mathcal{L}(\r_{AB})=\begin{pmatrix}
	0&0\\
	0&\mathcal{R}(\rho_{AB}-\r_A\otimes\r_B)
	\end{pmatrix},
	\end{align*}
	then if $\r_{AB}$ satisfies the enhanced realignment criterion, then
	\begin{align}
	&\norm{\mathcal{M}_{\a,\b}(\rho_{AB})}_{KF}\nonumber\\\le& \norm{\mathcal{M}_{\a,\b}(\rho_{A}\otimes\r_B)}_{KF}+\norm{\mathcal{R}(\rho_{AB}-\r_A\otimes\r_B)}_{KF}\nonumber\\
	\le& \sqrt{(\a^2+tr\rho_{A}^2)(\b^2+tr\rho_B^2)}+\sqrt{(1-tr\rho_A^2)(1-tr\rho_B^2)}\nonumber\\
	\le&\sqrt{(1+\a^2)(1+\b^2)}.
	\end{align}
	Here the last inequality is due to the Cauchy-Schwarz inequality. 
\par
Next we prove the converse $\Longleftarrow$:\par
First let us recall that when $A\in M_{m\times n}(\mathbb{R})$,
\begin{align}
\norm{A_{m\times n}}_{KF}=\max_{P\in \mathcal{U}_{m\times n}(\mathbb{R})} |TrA^{\dagger}P|,
\end{align}
where the maximum takes over all the unitary matrices $P\in\mathcal{U}_{m\times n}(\mathbb{R}) $. Assume $\r_{AB}$ is a separable state, according to the Theorem \ref{t1}, we have for any given $\a\ge0,\b\ge0$,
\begin{align*}
\norm{\mathcal{M}_{\a,\b}(\rho_{AB})}_{KF}\le \sqrt{(\a^2+1)(\b^2+1)},
\end{align*}
then let \begin{align*}
&f(\a,\b,P)\nonumber\\
=&\sqrt{(\a^2+1)(\b^2+1)}+\min_{P\in\mathcal{U}_{d_A^2\times d_B^2}}trP^{\dagger}\mathcal{M}_{\a\times\b}(\rho_{AB}),
\end{align*}we have
\begin{align}
&f(\a,\b,P)\nonumber\\=&\sqrt{(\a^2+1)(\b^2+1)}+\min_{P\in\mathcal{U}_{d_A^2\times d_B^2}}trP^{\dagger}\mathcal{M}_{\a\times\b}(\rho_{AB})\nonumber\\
=&\sqrt{(\a^2+1)(\b^2+1)}-\max_{P\in\mathcal{U}_{d_A^2\times d_B^2}}trP^{\dagger}\mathcal{M}_{\a\times \b}(\rho_{AB})\nonumber\\
=&\sqrt{(\a^2+1)(\b^2+1)}-\norm{\mathcal{M}_{\a,\b}(\r_{AB})}_{KF}
\ge0\label{t2},
\end{align}
here $P$ takes over all the unitary matrices.
\par As $\a,\b\in \mathbb{R},$ $\a,\b$ can be written as
\begin{align}
\a=r\cos\theta,\hspace{3mm}\b=r\sin\theta,
\end{align}
here $r\in\mathbb{R}^{+}$ and $\theta\in (0,\pi/2).$  
Next assume $P$ takes the following unitary matrix,
\begin{align*}
P=\begin{pmatrix}
-\sqrt{1-\frac{\eta^2}{r^2}}&\frac{\eta}{r}u^{T}\\
\frac{\eta}{r}v&\sqrt{1-\frac{\eta^2}{r^2}}O
\end{pmatrix},
\end{align*}
here $r$ tends to the infinity, $\eta$ is positive and infinitesimal with respect to r, $O$ is a $d_A^2\times d_B^2$ matrix. As $P$ is unitary, $\norm{u}=\norm{v}=1$,  $v=Ou,$ and $OO^{T}$ and $O^{T}O$ are projectors. Then the left hand side of $(\ref{t2})$ can be written as
\begin{widetext}
	\begin{align}
&f(\a,\b,P)\nonumber\\=	&\sqrt{(1+r^2\cos^2\theta)(1+r^2\sin^2\theta)}+\eta(\cos\theta vec(\rho_B)^{T}u
	+\sin\theta v^Tvec(\r_A))+\sqrt{1-\frac{\eta^2}{r^2}}(trR(\r_{AB})O^{T}-r^2\cos\theta\sin\theta)\nonumber\\
	=&\sqrt{(1+r^2\cos^2\theta)(1+r^2\sin^2\theta)}+trR(\r_{AB})O^{T}-r^2\cos\theta\sin\theta+\frac{\eta^2}{2}\cos\theta\sin\theta+\eta(\sin\theta vec(\rho_A)^{T}O+\cos\theta vec(\rho_B)^T)u+o(1),\nonumber\\
	\ge&\sqrt{(1+r^2\cos^2\theta)(1+r^2\sin^2\theta)}-r^2\cos\theta\sin\theta+\frac{\eta^2}{2}\cos\theta\sin\theta-\eta|\sin\theta vec(\rho_A)^{T}O+\cos\theta vec(\rho_B)^T|+trR(\r_{AB})O^{T}\nonumber\\
	=&\sqrt{(1+r^2\cos^2\theta)(1+r^2\sin^2\theta)}-r^2\cos\theta\sin\theta+trR(\r_{AB})O^{T}+\frac{\cos\theta\sin\theta}{2}[(\eta-m)^2-m^2],\label{t4}
	\end{align}
\end{widetext}
the second equality is due to that $\eta$ is infinitesimal relative to $r$. 
In the first inequality, when $u$ is antiparallel to -$(\sin\theta vec(\rho_A)^{T}O+\cos\theta vec(\rho_B)^T)^{T}$, the equality is valid. In the third equality, $m=\frac{|\sin\theta vec(\rho_A)^{T}O+\cos\theta vec(\rho_B)^T|}{\cos\theta\sin\theta}.$ When $\eta=m$, $(\ref{t4})$ gets the minimum.  Then
\begin{widetext}
	\begin{align*}
	&f(\a,\b,P)\\\ge&\sqrt{(1+r^2\cos^2\theta)(1+r^2\sin^2\theta)}-r^2\cos\theta\sin\theta+trR(\r_{AB})O^{T}-\frac{|\sin\theta vec(\rho_A)^{T}O+\cos\theta vec(\rho_B)^T|^2}{2\cos\theta\sin\theta}
	\\=&\frac{1}{2\sin\theta\cos\theta}-\frac{\sin^2\theta tr\rho_A^2+\cos^2\theta tr\rho_B^2+\cos\theta\sin\theta(vec(\rho_A)^{T}Ovec(\r_B)+vec(\rho_B)^TO^{T}vec(\rho_{A}))}{2\cos\theta\sin\theta}+trR(\rho_{AB})O^{T}\nonumber\\
	=&\frac{1}{2\sin\theta\cos\theta}-\frac{\tan\theta tr\rho_A^2}{2}-\frac{\cot\theta tr\rho_B^2}{2}+trR(\rho_{AB})O^{T}-\frac{vec(\rho_A)^{T}Ovec(\r_B)+vec(\rho_B)^TO^{T}vec(\rho_{A})}{2}\\
	\ge& \sqrt{(1-tr\rho_{A}^2)(1-tr\rho_B^2)}+trR(\rho_{AB})O^{T}-{vec(\rho_B)^TO^{T}vec(\rho_{A})},\nonumber\\
	=&\sqrt{(1-tr\rho_{A}^2)(1-tr\rho_B^2)}+tr(R(\rho_{AB}-\r_A\otimes\r_B))O^{T}.
	\end{align*}
\end{widetext}
Next as 

	\begin{align*}
	&\min_Otr(R(\rho_{AB}-\r_A\otimes\r_B))O^{T}\\=&-\max_Otr(R(\rho_{AB}-\r_A\otimes\r_B))O^{T}\\=&-\norm{(R(\rho_{AB}-\r_A\otimes\r_B))}_{KF}.
	\end{align*} 
	
\indent Hence there exists a $P\in M(\mathbb{R})$ such that $f(\a,\b,P)=\sqrt{(1-tr\rho_{A}^2)(1-tr\rho_B^2)}-\norm{R(\rho_{AB}-\r_A\otimes\r_B)}_{KF},$ that is, the dection power of the enhanced realignment criterion and Theorem \ref{t1} is equivalent for the bipartite states corresponds to a real matrix.\end{proof}\par
\section{Lower Bounds of Concurrence and CREN}
\indent In the bipartite entanglement theory, entanglement monotones are useful to distinguish the entanglement states and separable states. However, it is very hard to compute almost all entanglement monotones for arbitrary dimensional systems. In the last section, we first recall two popular entanglement monotones, concurrence and CREN, for mixed states of arbitrary bipartite systems, then we present a family of lower bounds of the two entanglement monotones. At last, by means of an example, our results can be used as separable criteria and can present better bounds than the results in \cite{de2007lower,li2020improved}. 
\par
Assume $\ket{\psi}_{AB}$ is a pure state in $\mathcal{H}_{AB}$, its concurrence is defined as
\begin{align*}
C(\ket{\psi}_{AB})=\sqrt{2(1-tr\rho_A^2)},
\end{align*}
here $\r_A=tr_B\ket{\psi}_{AB}\bra{\psi}.$ The concurrence for a mixed state $\r_{AB}$ is defined as
\begin{align*}
C(\rho_{AB})=\min_{\{p_i,\ket{\psi_i}\}}\sum_i p_iC(\ket{\psi_i}_{AB}),
\end{align*}
where the minimum takes over all the decompositions of $\r_{AB}=\sum_i p_i\ket{\psi_i}_{AB}\bra{\psi_i},$ $p_i\ge 0$ and $\sum_i p_i=1$. Next for a pure state $\ket{\psi}_{AB}$, its CREN \cite{lee2003convex} is defined as
\begin{align*}
\mathcal{N}(\ket{\psi})=\frac{\norm{(\ket{\psi}\bra{\psi})^{T_B}}-1}{k-1},
\end{align*}
here $k=\min(dim(\mathcal{H}_A),dim(\mathcal{H}_B))$, $(\ket{\psi}\bra{\psi})^{T_B}$ denotes the partial transpose of $\ket{\psi}\bra{\psi}.$ For a mixed state $\r_{AB},$ its CREN is defined as
\begin{align*}
C(\rho_{AB})=\min_{\{p_i,\ket{\psi_i}\}}\sum_i p_iC(\ket{\psi_i}_{AB}),
\end{align*}
where the minimum takes over all the decompositions of $\r_{AB}=\sum_i p_i\ket{\psi_i}_{AB}\bra{\psi_i},$ $p_i\ge 0$ and $\sum_i p_i=1$.\par 
Before presenting our results, we consider $\mathcal{M}_{\a,\b}(\ket{\psi})$ when $\ket{\psi}=\sum_{i=0}^{k-1} \sqrt{\lambda_i}\ket{ii}$ is a pure state, 
\begin{align*}
&\mathcal{M}_{\a,\b}(\ket{\psi}_{AB}\bra{\psi})=M_1+M_2,\\
&M_1=\begin{pmatrix}
\a\b& \a vec(\rho_B)^T\\
\b vec(\rho_A)&\Lambda_1
\end{pmatrix},M_2=\begin{pmatrix}
0& 0\\
0&\Lambda_2
\end{pmatrix},
\\
&vec(\rho_A)=(
\lambda_0,\underbrace{0,\cdots,0}_{k-1},0,\lambda_1,\underbrace{0,\cdots,0}_{k-1},\cdots,\lambda_{k-1}
)^T,\\
&vec(\rho_B)=(
\lambda_0,\underbrace{0,\cdots,0}_{k-1},0,\lambda_1,\underbrace{0,\cdots,0}_{k-1},\cdots,\lambda_{k-1}
)^T,\\
&\Lambda_1=\diag(\lambda_0,0,\cdots,0,0,\lambda_1,\cdots,0,\cdots,0,\cdots,\lambda_{k-1})\\
&\Lambda_2=\diag(0,\sqrt{\lambda_0\lambda_1},\cdots,\sqrt{\lambda_0\lambda_{k-1}},\\&\sqrt{\lambda_1\lambda_0},0,\cdots,\sqrt{\lambda_1\lambda_{k-1}},\cdots,\sqrt{\lambda_{k-1}\lambda_0},\cdots,0),
\end{align*}
As $trM_1M_2^{\dagger}=0,$ we have 
\begin{align}
\norm{\mathcal{M}_{\a,\b}(\ket{\psi}_{AB}\bra{\psi})}_{KF}=&\norm{M_1}_{KF}+\norm{M_2}_{KF}\nonumber\\
=&\norm{M_1}_{KF}+2\sum_{i<j}\sqrt{\lambda_i\lambda_j},
\end{align} 
In the following, we denote $\norm{\cdot}$ as $\norm{\cdot}_{KF}$. Next due to the definition of $\mathcal{M}_{\a,\b}(\cdot),$ $M_1=\mathcal{M}_{\a,\b}(\s_{AB}),$ $\s=\sum_i\lambda_i\ket{ii}\bra{ii}$ is a separable state. According to Theorem \ref{t1}, 
\begin{align}
\norm{\mathcal{M}_{\a,\b}(\ket{\psi}_{AB}\bra{\psi})}\le \sqrt{(1+\a^2)(1+\b^2)}+2\sum_{i<j}\sqrt{\lambda_i\lambda_j}.\label{t5}
\end{align}
Next in \cite{chen2005concurrence}, the authors showed that 
\begin{align*}
C^2(\ket{\psi})\ge\frac{8}{k(k-1)}(\sum_{i<j}\sqrt{\lambda_i\lambda_j})^2,
\end{align*}
then combining $(\ref{t5}),$ we have 
\begin{align*}
&C(\ket{\psi})\nonumber\\
\ge& \frac{\sqrt{2}}{\sqrt{k(k-1)}}(\norm{\mathcal{M}_{\a,\b}(\ket{\psi}_{AB}\bra{\psi})}-\sqrt{(1+\a^2)(1+\b^2)}).
\end{align*}
\par 
For a mixed state $\r_{AB},$ assume $\{p_i,\ket{\psi_i}\}$ is the optimal decomposition for $\rho_{AB}$ such that $C(\rho_{AB})=\sum_ip_iC(\ket{\psi_i}_{AB}),$ then
\begin{align}
C(\rho_{AB})=&\sum_ip_iC(\ket{\psi_i}_{AB})\nonumber\\
\ge&\frac{\sqrt{2}}{\sqrt{k(k-1)}}\sum_i p_i(\norm{\mathcal{M}_{\a,\b}(\ket{\psi_i}_{AB}\bra{\psi})}-\sqrt{(1+\a^2)(1+\b^2)})\nonumber\\
\ge&\frac{\sqrt{2}}{\sqrt{k(k-1)}}(\norm{\mathcal{M}_{\a,\b}(\rho_{AB})}-\sqrt{(1+\a^2)(1+\b^2)}). \label{c}
\end{align}
Then we have the following theorem,
\begin{Theorem}\label{t6}
Assume $\r_{AB}$ is a mixed state, for any $\a,\b\in\mathbb{R}^{+}$, we have
\begin{align*}
C(\rho_{AB})\ge \frac{\sqrt{2}}{\sqrt{k(k-1)}}(\norm{\mathcal{M}_{\a,\b}(\rho_{AB})}-\sqrt{(1+\a^2)(1+\b^2)}).
\end{align*}
\end{Theorem}
\par
Next we consider the lower bound for CREN of a mixed state $\r_{AB}.$ Assume $\ket{\psi}_{AB}=\sum_{i=0}^{k-1}\sqrt{\lambda}\ket{ii}$ is a pure state, then
\begin{align}
N(\ket{\psi}_{AB})=\frac{2(\sum_{j<i}\sqrt{\lambda_j\lambda_i})}{k-1},\label{t8}
\end{align} 
based on $(\ref{t5}),$ 
\begin{align*}
2\sum_{i<j}\sqrt{\lambda_i\lambda_j}\ge \norm{\mathcal{M}_{\a,\b}(\ket{\psi}_{AB}\bra{\psi})}- \sqrt{(1+\a^2)(1+\b^2)},
\end{align*}
let $\sum_ip_i\ket{\psi_i}\bra{\psi_i}$ be the optimal decomposition for $\r$ such that $\mathcal{N}(\r)=\sum_ip_iN(\ket{\psi_i})$ be the optimal for $\r$ such that $\mathcal{N}(\r)=\sum_ip_iN(\ket{\psi_i}),$ we have
\begin{align}
\mathcal{N}(\r_{AB})=&\sum_ip_i N(\ket{\psi_i})\nonumber\\
\ge&\sum_i p_i \frac{\norm{\mathcal{M}_{\a,\b}(\ket{\psi_i}\bra{\psi_i})}-\sqrt{(1+\a^2)(1+\b^2)}}{k-1}\nonumber\\
\ge&\frac{\norm{\mathcal{M}_{\a,\b}(\r_{AB})}-\sqrt{(1+\a^2)(1+\b^2)}}{k-1}.
\end{align}
Based on the above analytics, we have
\begin{Theorem}\label{t7}
	Assume $\r_{AB}$ is a mixed state, for any $\a,\b\in\mathbb{R}^{+}$, we have
	\begin{align}
	\mathcal{N}(\rho_{AB})\ge \frac{\norm{\mathcal{M}_{\a,\b}(\r_{AB})}-\sqrt{(1+\a^2)(1+\b^2)}}{k-1}.\label{f1}
	\end{align}
\end{Theorem}
\par 
Then we will present an example which shows our results are better than \cite{de2007lower}.
\begin{example}
	In this example, we consider the $3\otimes3$ PPT entangled state
	\begin{align*}
	\r=&\frac{1}{4}(I-\sum_i\ket{\psi_i}\bra{\psi_i}),\nonumber\\
	\ket{\psi_0}=&\frac{\ket{0}(\ket{0}-\ket{1})}{\sqrt{2}},\\
	\ket{\psi_1}=&\frac{(\ket{0}-\ket{1})\ket{2}}{\sqrt{2}},\\
	\ket{\psi_2}=&\frac{\ket{2}(\ket{1}-\ket{2})}{\sqrt{2}},\\\ket{\psi_3}=&\frac{(\ket{1}-\ket{2})\ket{0}}{\sqrt{2}},\\
	\ket{\psi_4}=&\frac{(\ket{0}+\ket{1}+\ket{2})(\ket{0}+\ket{1}+\ket{2})}{3}.
	\end{align*}
\par When choosing $\a=\b=1,$ according to Theorem \ref{t6}, $C(\rho_{AB})\ge 0.05399.$ By using the theorem 1 in \cite{de2007lower}, we have $C(\rho)\ge 0.052$, hence our bound is better than it.  In \cite{li2020improved}, the authors showed that $C(\r)\ge0.05554,$ here when we take $\a=\b=100$, the lower bound of $C(\r)$ is 0.055549.\par 
Next we consider a state by mixing $\rho$ with the white noise, 
\begin{align*}
\rho_p=pI/9+(1-p)\rho,
\end{align*}
here $I$ is the identity, $p\in[0,1].$ In Fig. \ref{fig1}, the orange line is the lower bound of $C(\r_p)$ obtained by Theorem 2 in \cite{de2007lower}. From this bound, $\r_p$ is entangled when $p\in [0,0.0507].$ By Theorem \ref{t6}, when taking $\a=\b=5,$ we obtain a lower bound of $C(\r_p)$ and plot it in Fig. \ref{fig1} as a blue line. There we have $\r_p$ is entangled when $p\in[0,0.1177]$. \par
\begin{figure}[t]
	\centering
	\includegraphics[width=90mm]{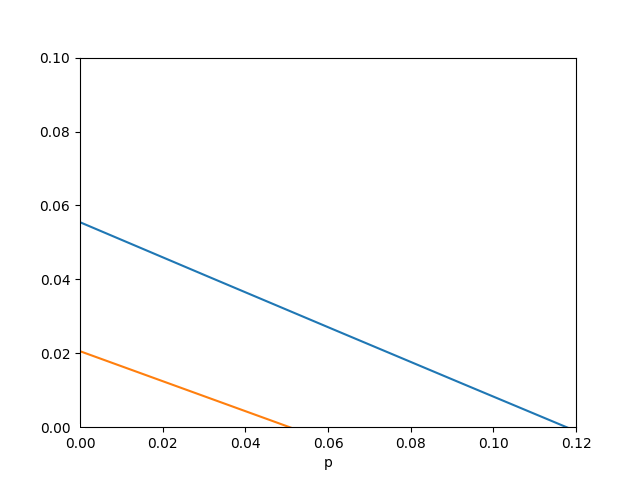}\\
	\caption{ Lower bound of $C(\r_p).$ The blue line is the bound given by Theorem \ref{t6}, while the orange line is the bound obtained by Theorem 2 in \cite{de2007lower}. }\label{fig1}
\end{figure}
In Fig. \ref{fig2}, we plot the bound of $\mathcal{N}(\r_p)$ obtained by Theorem \ref{t7}, there the orange line is on the bound of $\mathcal{N}(\r_p)$ with $\a=\b=7,$ the blue line is on the bound of $\mathcal{N}(\r_p)$ with $\a=\b=1.$
From the figure, we can see when taking $\a=\b=7,$ the lower bound is better.
\begin{figure}[b]
	\includegraphics[width=90mm]{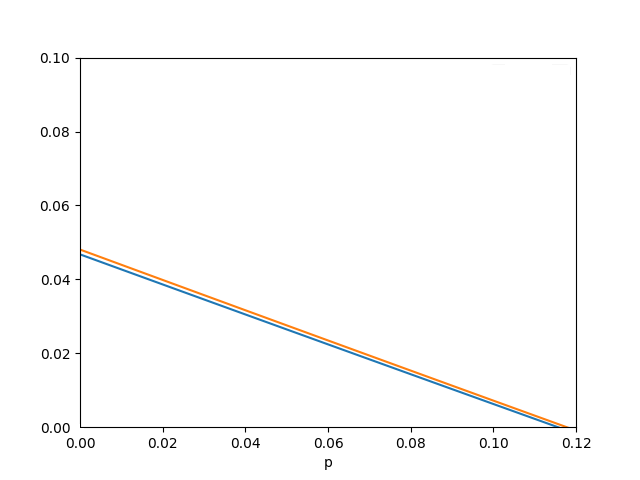}\\
	\caption{  Lower bound of CREN for the state $\r_p$. The orange line is the bound of $\mathcal{N}(\r_p)$ with $\a=\b=7,$ and the blue line is the bound with $\a=\b=1.$ }\label{fig2}
\end{figure}
\end{example}
\section{Conclusion} 
\indent To detect the entanglement of a bipartite state is essential in quantum entanglement theory. In this paper, we  presented a class of separability criteria which are better than CCNR criterion. Moreover, we proved that the detection power of our criteria is as strong as the enhanced realignment criterion for bipartite states corresponding to real density matrices. At last, we derived analytical lower bounds of the concurrence and CREN. Moreover, our methods here can also be used to obtain the lower bounds of some multipartite entanglement measures. We hope our work could shed some light on related studies.
\bibliographystyle{IEEEtran}
\bibliography{ref}

\begin{thebibliography}{10}
\providecommand{\url}[1]{#1}
\csname url@samestyle\endcsname
\providecommand{\newblock}{\relax}
\providecommand{\bibinfo}[2]{#2}
\providecommand{\BIBentrySTDinterwordspacing}{\spaceskip=0pt\relax}
\providecommand{\BIBentryALTinterwordstretchfactor}{4}
\providecommand{\BIBentryALTinterwordspacing}{\spaceskip=\fontdimen2\font plus
\BIBentryALTinterwordstretchfactor\fontdimen3\font minus
  \fontdimen4\font\relax}
\providecommand{\BIBforeignlanguage}[2]{{%
\expandafter\ifx\csname l@#1\endcsname\relax
\typeout{** WARNING: IEEEtran.bst: No hyphenation pattern has been}%
\typeout{** loaded for the language `#1'. Using the pattern for}%
\typeout{** the default language instead.}%
\else
\language=\csname l@#1\endcsname
\fi
#2}}
\providecommand{\BIBdecl}{\relax}
\BIBdecl

\bibitem{horodecki2009quantum}
R.~Horodecki, P.~Horodecki, M.~Horodecki, and K.~Horodecki, ``Quantum
  entanglement,'' \emph{Reviews of modern physics}, vol.~81, no.~2, p. 865,
  2009.

\bibitem{plenio2014introduction}
M.~B. Plenio and S.~S. Virmani, ``An introduction to entanglement theory,'' in
  \emph{Quantum Information and Coherence}.\hskip 1em plus 0.5em minus
  0.4em\relax Springer, 2014, pp. 173--209.

\bibitem{ekert1991quantum}
A.~K. Ekert, ``Quantum cryptography based on bell’s theorem,'' \emph{Physical
  review letters}, vol.~67, no.~6, p. 661, 1991.

\bibitem{bennett1993teleporting}
C.~H. Bennett, G.~Brassard, C.~Cr{\'e}peau, R.~Jozsa, A.~Peres, and W.~K.
  Wootters, ``Teleporting an unknown quantum state via dual classical and
  einstein-podolsky-rosen channels,'' \emph{Physical review letters}, vol.~70,
  no.~13, p. 1895, 1993.

\bibitem{bennett1992communication}
C.~H. Bennett and S.~J. Wiesner, ``Communication via one-and two-particle
  operators on einstein-podolsky-rosen states,'' \emph{Physical review
  letters}, vol.~69, no.~20, p. 2881, 1992.

\bibitem{peres1996separability}
A.~Peres, ``Separability criterion for density matrices,'' \emph{Physical
  Review Letters}, vol.~77, no.~8, p. 1413, 1996.

\bibitem{gurvits2003classical}
L.~Gurvits, ``Classical deterministic complexity of edmonds' problem and
  quantum entanglement,'' in \emph{Proceedings of the thirty-fifth annual ACM
  symposium on Theory of computing}, 2003, pp. 10--19.

\bibitem{rudolph2005further}
O.~Rudolph, ``Further results on the cross norm criterion for separability,''
  \emph{Quantum Information Processing}, vol.~4, no.~3, pp. 219--239, 2005.

\bibitem{chen2003matrix}
K.~Chen and L.-A. Wu, ``A matrix realignment method for recognizing
  entanglement,'' \emph{Quantum Inf. Comput.}

\bibitem{guhne2006entanglement}
O.~G{\"u}hne, M.~Mechler, G.~T{\'o}th, and P.~Adam, ``Entanglement criteria
  based on local uncertainty relations are strictly stronger than the
  computable cross norm criterion,'' \emph{Physical Review A}, vol.~74, no.~1,
  p. 010301, 2006.

\bibitem{de}
J.~I. De~Vicente, ``Separability criteria based on the bloch representation of
  densiity matrices,'' \emph{Quantum Inf. Comput.}

\bibitem{zhang2008entanglement}
C.-J. Zhang, Y.-S. Zhang, S.~Zhang, and G.-C. Guo, ``Entanglement detection
  beyond the computable cross-norm or realignment criterion,'' \emph{Physical
  Review A}, vol.~77, no.~6, p. 060301, 2008.

\bibitem{shen2015separability}
S.-Q. Shen, M.-Y. Wang, M.~Li, and S.-M. Fei, ``Separability criteria based on
  the realignment of density matrices and reduced density matrices,''
  \emph{Physical Review A}, vol.~92, no.~4, p. 042332, 2015.

\bibitem{shang2018enhanced}
J.~Shang, A.~Asadian, H.~Zhu, and O.~G{\"u}hne, ``Enhanced entanglement
  criterion via symmetric informationally complete measurements,''
  \emph{Physical Review A}, vol.~98, no.~2, p. 022309, 2018.

\bibitem{sarbicki2020family}
G.~Sarbicki, G.~Scala, and D.~Chru{\'s}ci{\'n}ski, ``Family of multipartite
  separability criteria based on a correlation tensor,'' \emph{Physical Review
  A}, vol. 101, no.~1, p. 012341, 2020.

\bibitem{sarbicki2020enhanced}
------, ``Enhanced realignment criterion vs linear entanglement witnesses,''
  \emph{Journal of Physics A: Mathematical and Theoretical}, vol.~53, no.~45,
  p. 455302, 2020.

\bibitem{wootters1998entanglement}
W.~K. Wootters, ``Entanglement of formation of an arbitrary state of two
  qubits,'' \emph{Physical Review Letters}, vol.~80, no.~10, p. 2245, 1998.

\bibitem{vedral1998entanglement}
V.~Vedral and M.~B. Plenio, ``Entanglement measures and purification
  procedures,'' \emph{Physical Review A}, vol.~57, no.~3, p. 1619, 1998.

\bibitem{vidal1999robustness}
G.~Vidal and R.~Tarrach, ``Robustness of entanglement,'' \emph{Physical Review
  A}, vol.~59, no.~1, p. 141, 1999.

\bibitem{terhal2000schmidt}
B.~M. Terhal and P.~Horodecki, ``Schmidt number for density matrices,''
  \emph{Physical Review A}, vol.~61, no.~4, p. 040301, 2000.

\bibitem{vidal2000entanglement}
G.~Vidal, ``Entanglement monotones,'' \emph{Journal of Modern Optics}, vol.~47,
  no. 2-3, pp. 355--376, 2000.

\bibitem{wei2003geometric}
T.-C. Wei and P.~M. Goldbart, ``Geometric measure of entanglement and
  applications to bipartite and multipartite quantum states,'' \emph{Physical
  Review A}, vol.~68, no.~4, p. 042307, 2003.

\bibitem{christandl2004squashed}
M.~Christandl and A.~Winter, ``“squashed entanglement”: an additive
  entanglement measure,'' \emph{Journal of mathematical physics}, vol.~45,
  no.~3, pp. 829--840, 2004.

\bibitem{lee2003convex}
S.~Lee, D.~P. Chi, S.~D. Oh, and J.~Kim, ``Convex-roof extended negativity as
  an entanglement measure for bipartite quantum systems,'' \emph{Physical
  Review A}, vol.~68, no.~6, p. 062304, 2003.

\bibitem{chen2005concurrence}
K.~Chen, S.~Albeverio, and S.-M. Fei, ``Concurrence of arbitrary dimensional
  bipartite quantum states,'' \emph{Physical review letters}, vol.~95, no.~4,
  p. 040504, 2005.

\bibitem{brandao2005quantifying}
F.~G. Brandao, ``Quantifying entanglement with witness operators,''
  \emph{Physical Review A}, vol.~72, no.~2, p. 022310, 2005.

\bibitem{de2007lower}
J.~I. de~Vicente, ``Lower bounds on concurrence and separability conditions,''
  \emph{Physical Review A}, vol.~75, no.~5, p. 052320, 2007.

\bibitem{Chen2012}
Z.-H. Chen, Z.-H. Ma, O.~Gühne, and S.~Severini, ``Estimating entanglement
  monotones with a generalization of the wootters formula,'' \emph{Physical
  Review Letters}, vol. 109, no.~20, p. 200503, 2012.

\bibitem{li2020improved}
M.~Li, Z.~Wang, J.~Wang, S.~Shen, and S.-m. Fei, ``Improved lower bounds of
  concurrence and convex-roof extended negativity based on bloch
  representations,'' \emph{Quantum Information Processing}, vol.~19, no.~4, pp.
  1--11, 2020.

\bibitem{bhatia2013matrix}
R.~Bhatia, \emph{Matrix analysis}.\hskip 1em plus 0.5em minus 0.4em\relax
  Springer Science \& Business Media, 2013, vol. 169.

\bibitem{bruss2000construction}
D.~Bru{\ss} and A.~Peres, ``Construction of quantum states with bound
  entanglement,'' \emph{Physical Review A}, vol.~61, no.~3, p. 030301, 2000.

\end{thebibliography}
\end{document}